\documentclass[twoside,onecolumn]{article}%
\usepackage{amssymb}
\usepackage{amsfonts}
\usepackage{amsmath}
\usepackage{graphicx}%
\providecommand{\U}[1]{\protect\rule{.1in}{.1in}}
\newtheorem{theorem}{Theorem}

\newtheorem{lemma}[theorem]{Lemma}

\newtheorem{remark}[theorem]{Remark}

\newenvironment{proof}[1][Proof]{\noindent\textbf{#1.} }{\ \rule{0.5em}{0.5em}}
\begin{document}

\title{\textbf{Self-Similar Solutions with Elliptic Symmetry for the Compressible
Euler and Navier-Stokes Equations in }$R^{N}$}
\author{M\textsc{anwai Yuen\thanks{E-mail address: nevetsyuen@hotmail.com }}\\\textit{Department of Applied Mathematics,}\\\textit{The Hong Kong Polytechnic University,}\\\textit{Hung Hom, Kowloon, Hong Kong}}
\date{Revised 19-Apr-2011v2}
\maketitle

\begin{abstract}
Based on Makino's solutions with radially symmetry, we extend the
corresponding ones with elliptic symmetry for the compressible Euler and
Navier-Stokes equations in $R^{N}$ ($N\geq2$). By the separation method, we
reduce the Euler and Navier-Stokes equations into $1+N$ differential
functional equations. In detail, the velocity is constructed by the novel
Emden dynamical system:%
\begin{equation}
\left\{
\begin{array}
[c]{c}%
\ddot{a}_{i}(t)=\frac{\xi}{a_{i}(t)\left(  \underset{k=1}{\overset{N}{\Pi}%
}a_{k}(t)\right)  ^{\gamma-1}}\text{, for }i=1,2,....,N\\
a_{i}(0)=a_{i0}>0,\text{ }\dot{a}_{i}(0)=a_{i1}%
\end{array}
\right.
\end{equation}
with arbitrary constants $\xi$, $a_{i0}$ and $a_{i1}$. Some blowup phenomena
or global existences of the solutions obtained could be shown.

MSC2010: 35B40, 35Q31, 35Q30, 37C10, 37C75 ,76N10

Key Words: Euler Equations, Navier-Stokes Equations, Analytical Solutions,
Elliptic Symmetry, Makino's Solutions, Self-Similar, Drift Phenomena, Emden
Equation, Blowup, Global Solutions

\end{abstract}

\section{Introduction}

The compressible Euler or Navier-Stokes equations are written as the follows:%
\begin{equation}
\left\{
\begin{array}
[c]{rl}%
{\normalsize \rho}_{t}{\normalsize +\nabla\cdot(\rho\vec{u})} &
{\normalsize =}{\normalsize 0}\\
\rho\left[  \vec{u}_{t}+\left(  \vec{u}\cdot\nabla\right)  \vec{u}\right]
+K\rho^{\gamma} & {\normalsize =}\mu\Delta\vec{u}%
\end{array}
\right.  \label{eq1cnsnsyy2011}%
\end{equation}
where the density $\rho=\rho(t,\vec{x})$ and velocity $\vec{u}=\vec{u}%
(t,\vec{x})=(u_{1},u_{2},....,u_{N})\in R^{N}$ with $\vec{x}=(x_{1}%
,x_{2},...,x_{N})\in R^{N}$. And $K>0$, $\gamma\geq1$ and $\mu\geq0$ are
constants. If $\mu=0$, the system (\ref{eq1cnsnsyy2011}) is the Euler
equations; if $\mu>0$, the system (\ref{eq1cnsnsyy2011}) is the Navier-Stokes equations.

The Euler and Navier-Stokes equations (\ref{eq1cnsnsyy2011}) are the very
fundamental models in fluid mechanics \cite{Lions} and \cite{CW}. Searching
particular solutions for the systems are the important part in mathematical
physics for understanding their nonlinear phenomena. By the separation method,
Makino firstly obtained the radial symmetry solutions for the Euler or
Navier-Stokes equations (\ref{eq1cnsnsyy2011}) in 1993
\cite{Makino93exactsolutions}. After that there are some other ways to
construct some particular solutions \cite{LW} and \cite{Yuen1dEuler} for these systems.

It is natural to seek solutions with elliptic symmetry for the Euler or
Navier-Stokes equations (\ref{eq1cnsnsyy2011}), based on the previous work. In
this brief article, we could generalize Makino's solutions to the
corresponding ones with elliptical symmetry and drift phenomena for these
systems in the following theorem:

\begin{theorem}
\label{thm:1cnsnsyy2011elliptic}To the Euler and Navier-Stokes equations
(\ref{eq1cnsnsyy2011}) in $R^{N}$, there exists a family of solutions:%
\begin{equation}
\left\{
\begin{array}
[c]{c}%
\rho=\frac{f(s)}{\underset{k=1}{\overset{N}{\Pi}}a_{k}}\\
u_{i}=\frac{\dot{a}_{i}}{a_{i}}\left(  x_{i}+d_{i}\right)  \text{ for
}i=1,2,....,N
\end{array}
\right.  \label{Yuensolcnsnsyy2011}%
\end{equation}
where%
\begin{equation}
f(s)=\left\{
\begin{array}
[c]{c}%
\alpha e^{-\frac{\xi}{2K}s}\text{
\ \ \ \ \ \ \ \ \ \ \ \ \ \ \ \ \ \ \ \ \ \ \ \ \ \ \ \ \ for }\gamma=1\\
\max\left(  \left(  -\frac{\xi(\gamma-1)}{2K\gamma}s+\alpha\right)  ^{\frac
{1}{\gamma-1}},\text{ }0\right)  \text{ for }\gamma>1
\end{array}
\right.
\end{equation}
with $s=\underset{k=1}{\overset{N}{\sum}}\frac{(x_{k}+d_{k})^{2}}{a_{k}%
(t)^{2}}$, arbitrary constants $\alpha\geq0,$ $d_{k}$ and $\xi$;\newline and
the auxiliary functions $a_{i}=a_{i}(t)$ satisfy the Emden dynamical system:
\begin{equation}
\left\{
\begin{array}
[c]{c}%
\ddot{a}_{i}=\frac{\xi}{a_{i}\left(  \underset{k=1}{\overset{N}{\Pi}}%
a_{k}\right)  ^{\gamma-1}}\text{, for }i=1,2,....,N\\
a_{i}(0)=a_{i0}>0,\text{ }\dot{a}_{i}(0)=a_{i1}%
\end{array}
\right.  \label{Emdengeneral1}%
\end{equation}
with arbitrary constants $a_{i0}$ and $a_{i1}.$\newline In particular, with
$\gamma=1,$\newline(1a) for $\xi<0$, the solutions (\ref{Yuensolcnsnsyy2011})
blow up on a finite time;\newline(1b) for $\xi>0,$ the solutions
(\ref{Yuensolcnsnsyy2011}) exists globally.\newline with $\gamma>1,$%
\newline(2a) for $\xi<0$ and some $a_{i1}<0$, the solutions
(\ref{Yuensolcnsnsyy2011}) blow up on or before the finite time
\begin{equation}
T=\min(-a_{i0}/a_{i1}:a_{1i}<0,i=1,2,...,N);
\end{equation}
(2b) for $\xi>0$ and $a_{i1}\geq0$ the solutions (\ref{Yuensolcnsnsyy2011})
exist globally.
\end{theorem}

\begin{remark}
When $a_{1}=a_{2}=....=a_{N}=a(t)$, the solutions are with radial symmetry and
the Emden dynamical system (\ref{Emdengeneral1}) returns to the conventional
Emden equation:%
\begin{equation}
\left\{
\begin{array}
[c]{c}%
\ddot{a}(t)=\frac{\xi}{a(t)^{N(\gamma-1)+1}}\\
a(0)=a_{0}>0,\text{ }\dot{a}(0)=a_{1}.
\end{array}
\right.  \label{Emdeneqeq1}%
\end{equation}
This class of analytical solutions (\ref{Yuensolcnsnsyy2011}) with radial
symmetry for the compressible Euler equations (\ref{eq1cnsnsyy2011}) was first
discovered by Makino in \cite{Makino93exactsolutions}. Otherwise, the
solutions (\ref{Yuensolcnsnsyy2011}) are with elliptical symmetry for $N\geq2$.
\end{remark}

\section{The Separation Method}

Very recently, Yeung and Yuen in \cite{YYCNSNS2011} discovered the implicit or
explicit functions for the mass equations (\ref{Yuensolcnsnsyy2011})$_{1}$. In
this section, we apply their result in the explicit expression to have the
following lemma:

\begin{lemma}
[Lemma 1 in \cite{YYCNSNS2011}]\label{lem:generalsolutionformasseq}For the
equation of conservation of mass:
\begin{equation}
\rho_{t}+\nabla\cdot\left(  \rho\vec{u}\right)  =0,
\label{massequationsphericalcnsnsyy2011}%
\end{equation}
there exist solutions,%
\begin{equation}
\left\{
\begin{array}
[c]{c}%
\rho=\frac{f\left(  \frac{x_{1}+d_{1}}{a_{1}(t)},\frac{x_{2}+d_{2}}{a_{2}%
(t)},....,\frac{x_{N}+d_{N}}{a_{N}(t)}\right)  }{\underset{i=1}{\overset
{N}{\Pi}}a_{i}(t)}\\
u_{i}=\frac{\dot{a}_{i}(t)}{a_{i}(t)}\left(  x_{i}+d_{i}\right)  \text{ for
}i=1,2,....,N
\end{array}
\right.  \label{generalsolutionformassequationcnsnsyy2011}%
\end{equation}
with an arbitrary $C^{1}$ function $f\geq0$ and $a_{i}(t)>0$ and constants
$d_{i}$.
\end{lemma}

For better understanding the lemma, the proof is also provided here.

\begin{proof}
We plug the functions
\begin{equation}
\left\{
\begin{array}
[c]{c}%
\rho=\rho(t,\vec{x})\\
u_{i}=\frac{\dot{a}_{i}(t)}{a_{i}(t)}\left(  x_{i}+d_{i}\right)  \text{ for
}i=1,2,....,N,
\end{array}
\right.
\end{equation}
into the mass equation (\ref{massequationsphericalcnsnsyy2011}):%
\begin{equation}
\rho_{t}+\nabla\rho\cdot\vec{u}+\rho\nabla\cdot\vec{u}=0
\end{equation}%
\begin{equation}
\frac{\partial}{\partial t}\rho+\underset{i=1}{\overset{N}{\sum}}%
\frac{\partial}{\partial x_{i}}\rho\frac{\dot{a}_{i}(t)}{a_{i}(t)}(x_{i}%
+d_{i})+\underset{i=1}{\overset{N}{\sum}}\frac{\rho\dot{a}_{i}(t)}{a_{i}%
(t)}=0. \label{yylinearpdecnsnsyy2011}%
\end{equation}
The general solutions for the semi-linear partial differential equation
(\ref{yylinearpdecnsnsyy2011}) are:%
\begin{equation}
F\left(  \underset{i=1}{\overset{N}{\Pi}}a_{i}(t)\rho,\frac{x_{1}+d_{1}}%
{a_{1}(t)},\frac{x_{2}+d_{2}}{a_{2}(t)},....,\frac{x_{N}+d_{N}}{a_{N}%
(t)}\right)  =0 \label{YYgeneralsolutioncnsnsyy2011}%
\end{equation}
with an arbitrary $C^{1}$ function $F$ such that $\rho\geq0$.\newline We take
the explicit one as
\begin{equation}
\rho=\frac{f\left(  \frac{x_{1}+d_{1}}{a_{1}(t)},\frac{x_{2}+d_{2}}{a_{2}%
(t)},....,\frac{x_{N}+d_{N}}{a_{N}(t)}\right)  }{\underset{i=1}{\overset
{N}{\Pi}}a_{i}(t)}.
\end{equation}
Therefore, the proof is completed.
\end{proof}

The following proof is the checking our constructed functions
(\ref{Yuensolcnsnsyy2011}) for the main result:

\begin{proof}
[Proof of Theorem \ref{thm:1cnsnsyy2011elliptic}]We observe that the functions
(\ref{Yuensolcnsnsyy2011}) satisfy the conditions in Lemma
\ref{lem:generalsolutionformasseq} for the mass equation (\ref{eq1cnsnsyy2011}%
)$_{1}$. Alternatively, readers can plug the functions
(\ref{Yuensolcnsnsyy2011}) to balance the mass equation by the directional
computation like in \cite{Makino93exactsolutions} and \cite{YUENJMP2008}.

For the $i$-th momentum equation (\ref{eq1cnsnsyy2011})$_{2}$ of the Euler and
Navier-Stokes equations, we define the self-similar variable with elliptic
symmetry (for not $a_{1}=a_{2}=....=a_{N}$) and drift phenomena (for not all
$d_{i}=0$):
\begin{equation}
s=\underset{k=1}{\overset{N}{\sum}}\frac{(x_{k}+d_{k})^{2}}{a_{k}(t)^{2}}%
\end{equation}
to have:%
\begin{align}
&  \rho\left[  \frac{\partial u_{i}}{\partial t}+\sum_{k=1}^{N}u_{k}%
\frac{\partial u_{i}}{\partial x_{k}}\right]  +K\frac{\partial}{\partial
x_{i}}\rho^{\gamma}+\mu\Delta u_{i}\\
&  =\rho\left[  \frac{\partial}{\partial t}\left(  \frac{\dot{a}_{i}}{a_{i}%
}(x_{i}+d_{i})\right)  +\left(  \frac{\dot{a}_{i}}{a_{i}}(x_{i}+d_{i})\right)
\frac{\partial}{\partial x_{i}}\left(  \frac{\dot{a}_{i}}{a_{i}}(x_{i}%
+d_{i})\right)  \right]  +K\gamma\rho^{\gamma-1}\frac{\partial}{\partial
x_{i}}\rho\\
&  =\rho\left\{  \left[  \left(  \frac{\ddot{a}_{i}}{a_{i}}-\frac{\left(
\dot{a}_{i}\right)  ^{2}}{\left(  a_{i}\right)  ^{2}}\right)  (x_{i}%
+d_{i})+\frac{\left(  \dot{a}_{i}\right)  ^{2}}{\left(  a_{i}\right)  ^{2}%
}(x_{i}+d_{i})\right]  +K\gamma\rho^{\gamma-2}\frac{\partial}{\partial x_{i}%
}\frac{f(s)}{\underset{k=1}{\overset{N}{\Pi}}a_{k}}\right\} \\
&  =\rho\left\{  \frac{\ddot{a}_{i}}{a_{i}}(x_{i}+d_{i})+2K\gamma
\frac{f(s)^{\gamma-2}}{\left(  \underset{k=1}{\overset{N}{\Pi}}a_{k}\right)
^{\gamma-2}}\frac{\dot{f}\left(  s\right)  }{\left(  \underset{k=1}%
{\overset{N}{\Pi}}a_{k}\right)  }\left(  \frac{x_{i}+d_{i}}{a_{i}^{2}}\right)
\right\} \\
&  =\frac{\left(  x_{i}+d_{i}\right)  \rho}{a_{i}^{2}}\left\{  \ddot{a}%
_{i}a_{i}+2K\gamma\frac{f(s)^{\gamma-2}\dot{f}\left(  s\right)  }{\left(
\underset{k=1}{\overset{N}{\Pi}}a_{k}\right)  ^{\gamma-1}}\right\} \\
&  =\frac{\left(  x_{i}+d_{i}\right)  \rho}{a_{i}^{2}\left(  \underset
{k=1}{\overset{N}{\Pi}}a_{k}\right)  ^{\gamma-1}}\left\{  \xi+2K\gamma
f(s)^{\gamma-2}\dot{f}\left(  s\right)  \right\}
\end{align}
with the $N$-dimensional Emden dynamical system
\begin{equation}
\left\{
\begin{array}
[c]{c}%
\ddot{a}_{i}(t)=\frac{\xi}{a_{i}(t)\left(  \underset{k=1}{\overset{N}{\Pi}%
}a_{k}(t)\right)  ^{\gamma-1}}\text{ for }i=1,2,...,N\\
a_{i}(0)=a_{i0}>0,\text{ }\dot{a}_{i}(0)=a_{i1}%
\end{array}
\right.  \label{EmdenEmden}%
\end{equation}
with arbitrary constants $\xi$, $a_{i0}$ and $a_{i1}.$\newline Here, the local
existence for the Emden dynamical system (\ref{EmdenEmden}) can be guaranteed
by the fixed point theorem. Then, we further require the first order ordinary
differential equation:%
\begin{equation}
\left\{
\begin{array}
[c]{c}%
\frac{\xi}{2K\gamma}+f(s)^{\gamma-2}\dot{f}\left(  s\right)  =0\\
f(0)=\alpha\geq0.
\end{array}
\right.  \label{firstODE}%
\end{equation}
The above equation (\ref{firstODE}) can be solved exactly by\newline%
\begin{equation}
f(s)=\left\{
\begin{array}
[c]{c}%
\alpha e^{-\frac{\xi}{2K}s}\text{ }%
\ \ \ \ \ \ \ \ \ \ \ \ \ \ \ \ \ \ \ \ \ \ \ \ \ \ \ \ \ \ \text{for }%
\gamma=1\\
\max\left(  \left(  -\frac{\xi(\gamma-1)}{2K\gamma}s+\alpha\right)  ^{\frac
{1}{\gamma-1}},\text{ }0\right)  \text{ for }\gamma>1.
\end{array}
\right.
\end{equation}
Thus, the functions (\ref{Yuensolcnsnsyy2011}) are the solutions for the Euler
and Navier-Stokes equations (\ref{eq1cnsnsyy2011}).\newline In particular,
with $\gamma=1,$ as the Emden dynamical system (\ref{EmdenEmden}) becomes to
be the conventional Emden equation:%
\begin{equation}
\left\{
\begin{array}
[c]{c}%
\ddot{a}_{i}(t)=\frac{\xi}{a_{i}(t)}\\
a_{i}(0)=a_{i0}>0,\text{ }\dot{a}_{i}(0)=a_{i1},
\end{array}
\right.
\end{equation}
we may use the energy method in classical mechanics (or readers may refer the
lemma 7\ in \cite{YJMAA2008}) to show that (1a) for $\xi<0,$ functions
$a_{i}(t)$ blow up on a finite time;\newline(1b) for $\xi>0$, the functions
$a_{i}(t)$ exist globally.\newline With $\gamma>1$,\newline(2a) for $\xi<0$
and some $a_{i1}<0$, by comparing the second order linear ordinary
differential equations:%
\begin{equation}
\left\{
\begin{array}
[c]{c}%
\ddot{a}_{i}\leq0\\
a_{i}(0)=a_{i0}>0,\text{ }\dot{a}_{i}(0)=a_{i1},
\end{array}
\right.
\end{equation}
we can show that the solutions (\ref{Yuensolcnsnsyy2011}) blow up on or before
the finite time
\begin{equation}
T=\min(-a_{i0}/a_{i1}:a_{i1}<0,\text{ }i=1,2,...,N);
\end{equation}
(2b) for $\xi>0$ and all $a_{i1}\geq0$, similarly, it is clear for that the
solutions (\ref{Yuensolcnsnsyy2011}) exist globally.\newline The proof is completed.
\end{proof}

\section{Conclusion and Discussion}

In this brief paper, the analytically self-similar solutions
(\ref{Yuensolcnsnsyy2011}) with elliptic symmetry and drift phenomena for the
compressible Euler and Navier-Stokes equations in $R^{N}$ ($N\geq2$) are
constructed by the separation method. We reduce the Euler and Navier-Stokes
equations (\ref{eq1cnsnsyy2011}) into the $1+N$ differential functional
equations:
\begin{equation}
(f(s),a_{i}(t)\text{ for }i=1,2,....N).
\end{equation}
In addition, by analyzing the Emden dynamical system (\ref{Emdengeneral1}):%
\begin{equation}
\left\{
\begin{array}
[c]{c}%
\ddot{a}_{i}(t)=\frac{\xi}{a_{i}(t)\left(  \underset{k=1}{\overset{N}{\Pi}%
}a_{k}(t)\right)  ^{\gamma-1}}\text{, for }i=1,2,....,N\\
a_{i}(0)=a_{i0}>0,\text{ }\dot{a}_{i}(0)=a_{i1},
\end{array}
\right.  \label{Emdengeneral 2}%
\end{equation}
some blowup or global properties of the constructed solutions
(\ref{Yuensolcnsnsyy2011}) can be shown easily.

However, for the $N$-dimensional $(N\geq2)$ Emden dynamical system
(\ref{Emdengeneral 2}) with arbitrary constants $\xi$, $\gamma$, $a_{i0}$ and
$a_{i1}$, the all blowup sets, blowup times and asymptotic analysis of the
solutions are not clear to be obtained. Computing simulation and rigorous
mathematical proofs for the system (\ref{Emdengeneral 2}) can be followed to
understand the evolution of the constructed flows (\ref{Yuensolcnsnsyy2011})
for the Euler and Navier-Stokes equations in the future.


\begin{thebibliography}{9}                                                                                                %


\bibitem {CW}Chen G.Q. and Wang D.H. (2002), \textit{The Cauchy Problem for
the Euler Equations for Compressible Fluids}, Handbook of Mathematical Fluid
Dynamics \textbf{I}, 421--543, North-Holland, Amsterdam.

\bibitem {Lions}Lions P.L. (1998), Mathematical Topics in Fluid Mechanics
\textbf{2,}. Compressible Models, Oxford Lecture Series in Mathematics and its
Applications \textbf{10}, Oxford: Clarendon Press.

\bibitem {LW}Li, T.H. and Wang, D.H. (2006), \textit{Blowup Phenomena of
Solutions to the Euler Equations for Compressible Fluid Flow}, J. Differential
Equations \textbf{221}, 91--101.

\bibitem {Makino93exactsolutions}Makino T. (1993), \textit{Exact Solutions for
the Compressible Eluer Equation}, Journal of Osaka Sangyo University Natural
Sciences \textbf{95}, 21--35.

\bibitem {YYCNSNS2011}Yeung L.H. and Yuen M.W. (2011), \textit{Note for "Some
Exact Blowup Solutions to the Pressureless Euler Equations in }$R^{N}%
$\textit{" [Commun. Nonlinear Sci. Numer. Simul. 16 (2011), 2993-2998]},
Commun. Nonlinear Sci. Numer. Simul., In Press, DOI: 10.1016/j.cnsns.2011.04.016.

\bibitem {YJMAA2008}Yuen M.W. (2008a), \textit{Analytical Blowup Solutions to
the 2-dimensional Isothermal Euler-Poisson Equations of Gaseous Stars}, J.
Math. Anal. Appl. \textbf{341 (}2008\textbf{), }445--456.

\bibitem {YUENJMP2008}Yuen M.W. (2008b), \textit{Analyitcal Solutions to the
Navier-Stokes Equations}, J. Math. Phys. \textbf{49}, 113102, 10pp.

\bibitem {Yuen1dEuler}Yuen M.W., \textit{Perturbational Blowup Solutions to
the 1-dimensional Compressible Euler Equations}, Pre-print, arXiv:1012.2033.
\end{thebibliography}
\end{document}